\RequirePackage{fix-cm}
\documentclass[smallextended]{svjour3}       
\smartqed  
\usepackage{graphicx}

\usepackage{amsmath}		
\usepackage{amssymb}		
\usepackage{tocloft}		
\usepackage{float}			
\usepackage{graphicx}		
\usepackage{setspace}		

\usepackage[utf8]{inputenc}
\usepackage{graphicx}
\usepackage{algorithm}
\usepackage[noend]{algpseudocode}
\usepackage{color}
\usepackage{tcolorbox}
\usepackage{amsmath}
\usepackage{float}
\usepackage{gensymb}
\usepackage{makeidx}  
\usepackage{amssymb}
\usepackage{graphicx}
\usepackage{mathrsfs}

\usepackage{mathrsfs}
\usepackage{pgf}
\usepackage{multirow,array}
\usepackage{mathpartir}
\usepackage{float}
\usepackage{tikz}
\usetikzlibrary{arrows,automata,positioning,shapes.multipart}
\usepackage{stmaryrd}

\usepackage{fancyvrb}
\usepackage{alltt}
\usepackage{listings}
\usepackage[numbers,sort&compress]{natbib}
\usepackage{booktabs} 
\setcitestyle{max-names=1}

\newcommand{\SFive}{$\mathcal{S}$\textit{5}}
\newcommand{\KDFourFive}{$\mathcal{KD}$\textit{45}}
\newcommand{\LSED}{$\mathcal{LSED}^R$}
\newcommand{\LSEDmin}{$\mathcal{LSED}^S$}
\newcommand{\Rel}[1]{R_{#1}}

\newcommand{\Kns}[1]{\mathbf{K}_{\mathbf{#1}}\,} 
 
\newcommand{\Bels}[1]{\mathbf{B}_{\mathbf{#1}}\,} 
 
\newcommand{\Poss}[1]  {\langle \mathbf{K}_{#1} \rangle\,}

\newcommand{\BPoss}[1]{\langle \mathbf{B}_{#1} \rangle\,}

\newcommand{\etal}{\textit{et.~al.}}

\newcommand{\iiff}{{\  \Leftrightarrow \ }}

\newcommand {\tland}{{\ \wedge \ }}
\newcommand {\tlor}{{\ \vee \ }}
\newcommand {\tlnot}{{\neg}}
\newcommand {\iimplies}{{\ \Rightarrow \ }}

\begin{document}

\title{L\"ob-Safe Logics for Reflective Agents
}


\author{Seth Ahrenbach         \and
        Second Author 
}


\institute{Seth Ahrenbach \at
              \email{ahrenbach.seth@gmail.com}           
           \and
           S. Author \at
              second address
}

\date{Received: date / Accepted: date}

\maketitle

\begin{abstract}
Epistemic and doxastic logics are modal logics for knowledge and belief, and serve as foundational models for rational agents in game theory, philosophy, and computer science. We examine the consequences of modeling agents capable of a certain sort of reflection. Such agents face a formal difficulty due to L\"ob's Theorem, called L\"ob's Obstacle in the literature. We show how the most popular axiom schemes of epistemic and doxastic logics suffer from L\"ob's Obstacle, and present two axiom schemes that that avoid L\"ob's Obstacle, which we call Reasonable L\"ob-Safe Epistemic Doxastic logic (\LSED) and Supported L\"ob-Safe Epistemic Doxastic logic (\LSEDmin). 
\end{abstract}
\section{Introduction}

The standard formalization for an agent's belief is the modal logic $\mathit{KD}45$. Similarly, the standard for an agent's knowledge is the modal logic $\mathit{S5}$. The logics of belief and knowledge in the literature ignore a problem facing agents of a certain reflective type, identified by Smullyan in \cite{smullyan}, who can reason about self-referential sentences.\footnote{Smullyan refers to these as reflexive reasoners, but we use the term `reflective' in order to avoid ambiguity with the reflexive frame condition on worlds.} This paper confronts the problem of reflective reasoners facing modal logics of belief and knowledge, showing that most of the standard approaches fall short. We identify candidate multimodal logics of knowledge and belief that avoid the problem for reflective reasoning agents, and explain the different attributes of agents modeled by each.

In Smullyan's, ``Logicians who reason about themselves," he considers epistemic problems related to undecidability results in mathematics. He identifies ``a complete parallelism between logicians who believe propositions and mathematical systems that prove propositions." In provability logic, the formula $\varphi \iiff \tlnot \Bels{i}\varphi$ expresses the G\"odel proposition, ``This proposition is not provable in system $i$."\footnote{We ignore subtleties of encodings, here, because they represent the same proposition whether G\"odel-encoded or not.} In a doxastic interpretation, the same formula expresses the reflective belief, ``agent $i$ does not believe this proposition." This means that for any doxastic or epistemic system, if the agents it models are reflective, then it must properly handle the complications that arise from such self-reference.  

A \emph{reflective} agent is one that can form beliefs and knowledge about self-referential sentences and propositions. A biconditional is used to formalize the self-reference, as in the following: $\varphi \iiff (\Bels{i}\varphi \iimplies \psi)$. The reflective proposition ``agent $i$ does not believe this proposition" is of that form, where $\psi$ is replaced with $\bot$ in order to formalize it: $\varphi \iiff (\Bels{i}\varphi \iimplies \bot)$, or equivalently, $\varphi \iiff (\tlnot\Bels{i}\varphi)$. 

The above uses the $\Bels{i}$ operator for belief, but just as easily we could have used $\Kns{i}$. Most humans are capable of reasoning about self-referential propositions, so a logic for human knowledge ought to include such propositions. However, with the Truth Axiom of the knowledge operator, this seems to yield inconsistency. Returning to Smullyan's parallelism, the Truth Axiom, $\Kns{i}\varphi \iimplies \varphi$ would translate to provability logic as the axiom ``if formula $\varphi$ is provable in system $i$, then $\varphi$ is true". This cannot be an axiom in provability logic, as it is for epistemic logic, because it yields the very same inconsistency that mathematical system $i$ would face if it could prove its own soundness, due to G\"odel's Second Incompleteness Theorem.\cite{godel}

The problem facing epistemic and doxastic logics is due to L\"ob in \cite{Lob} from the mathematical logic perspective, and Smullyan presents it in \cite{smullyan} from the doxastic logic perspective. It intersects with contemporary research in artificial intelligence foundations, as in \cite{yudkowsky}, which addresses the problem as it pertains to agents being confident in their own conclusions. They have named the problem L\"ob's Obstacle, or the L\"obian Obstacle.

In what follows, we examine the mathematical logic that underpins this issue. Just as Smullyan identified the problem as one facing reflective reasoners with Axioms K and 4 and the Rule of Necessitation in \cite{smullyan}, L\"ob identified these same conditions as the ones that allow a mathematical system to derive his theorem in \cite{Lob}. We highlight the obstacle this presents to epistemic and doxastic logics, and show how to avoid it. We present relaxed axiom schemas for reasoning about knowledge and belief for reflective agents.


\section{L\"ob's Theorem}
\label{sec:lob_section}
L\"ob's Theorem takes its name from Martin Hugo L\"ob, who tackled a question of mathematical logic posed by Leon Henkin in the years following the results of G\"odel. Henkin asked what could be said of propositions asserting their own provability, as opposed to unprovability in the case of G\"odel sentences.\cite{sep_prov_log} L\"ob answered by showing that in a consistent system, proof of soundness is limited to propositions that are actually provable, and not as a general property of the system.

L\"ob's Theorem in provability logic is,
\begin{center}~\label{lob}
	\begin{equation}
	\Box(\Box\varphi \iimplies \varphi)\iimplies \Box\varphi.
	\end{equation}
\end{center}
The $\Box$ is interpreted as ``provability" in some formal system at least as powerful as Peano arithmetic. However, the theorem will occur in other modal logics if certain conditions are met. For example, if the $\Box$ is interpreted as knowledge, and those conditions are met, then L\"ob's Theorem will hold.




Here we give a template derivation of L\"ob's Theorem, which we shall refer to below when describing how L\"ob's Obstacle corrupts various epistemic logics. Following this proof, we present the conditions that will cause a modal logic axiom schema to derive the theorem.
\\
\textbf{Proof:}

\begin{tabular}{ll}
\toprule
(1) & $\Box(\Box\varphi \rightarrow \varphi)$ \dotfill Assumption \\
(2) & $\Box(\psi \leftrightarrow (\Box\psi\rightarrow \varphi))$ \dotfill L\"ob Sentence\footnote{Sometimes referred to as a Curry sentence after logician Haskell Curry.} \\
(3) & $\Box(\Box\psi \leftrightarrow \Box(\Box\psi \rightarrow \varphi))$ \dotfill Axiom K, (2) \\
(4) & $\Box(\Box\psi \rightarrow \Box(\Box\psi \rightarrow \varphi))$ \dotfill (3) Simplification of $\leftrightarrow$ \\
(5) & $\Box(\Box\psi \rightarrow (\Box\Box\psi \rightarrow \Box\varphi))$ \dotfill (4) Axiom K \\
(6) & $\Box(\Box\psi \rightarrow \Box\Box\psi)$ \dotfill Axiom 4 \\
(7) & $\Box(\Box\psi \rightarrow \Box\varphi)$ \dotfill (5), (6) \\
(8) & $\Box(\Box\psi \rightarrow \varphi)$ \dotfill (7), (1) \\
(9) & $\Box\psi$ \dotfill (3), (8) \\
(10) & $\Box\Box\psi$ \dotfill (9), Axiom 4 \\
(11) & $\Box\Box\psi \rightarrow \Box\varphi$ \dotfill (8), Axiom K \\
(12) & $\Box\varphi$ \dotfill (10), (11) \\
\bottomrule
\end{tabular}

\hfill $\mathcal{QED}$

Mathematical and provability logicians refer to the key components of this proof as L\"ob Conditions\cite{Boolos}. Identifying them in the proof above helps us identify which epistemic logics collide with L\"ob's Obstacle. Conversely, understanding how the L\"ob Conditions interact helps us construct epistemic logics that avoid L\"ob's Obstacle.

The Conditions are:
\begin{enumerate}
	\item The L\"ob Sentence. A self-referential or reflective sentence, also formalizable as a modal fixed point.
	\item Axiom K. The standard distribution axiom of normal modal logics.
	\item Axiom 4. The axiom corresponding to a transitive frame relation.
	\item The rule of necessitation. Likewise a standard feature of normal modal logics.
\end{enumerate} 

The L\"ob Sentence is sometimes not mentioned as a Condition, because L\"ob's Theorem is typically studied in the context of mathematical logic or provability logic, where such self-referential expressiveness is known to exist. We point out, however, that humans are capable of reasoning about self-referential sentences, and any advanced artificial agent will be able to do so, as well. Because systems for representing human-like reasoners should include self-referential sentences and modal fixed points, this condition is satisfied for our concerns. For example, in the foundations of game theory, ideally rational agents can reason about common knowledge among each other, which is itself defined as a modal fixed point, as showin in Barwise \cite{barwise1988three}. Items (2) and (4) are constants for all normal modal logics. What remains is for an axiom schema to include (3).

Finally, we note the importance of L\"ob's Theorem's antecedent: $\Box(\Box\varphi \iimplies \varphi)$. Epistemic logics typically include the antecedent as an axiom, representing the widely held view that knowledge entails truth, in which case L\"ob's Theorem will allow us to derive $\Box\varphi$ for all $\varphi$. We assert that epistemic logicians must come to terms with L\"ob's Obstacle and ensure that their models of human-level knowledge do not crash into it.

We define a $\mathit{crash}$ into L\"ob's Obstacle as follows:
\begin{definition}\label{Crash}
	A logic $\mathcal{L}$ with modal operators $\mathbf{\Box_{{i}\in\mathbb{N}}}$ $\mathbf{crashes}$ just in case the axioms of $\mathcal{L}$ derive L\"ob's Theorem for at least one $\Box_{i}$ and $\mathcal{R}_i$, the relation defining $\Box_{i}$, is either serial or reflexive.
	
\end{definition}

It suffices to say that the relevant relation $\mathcal{R}_i$ be $\mathit{serial}$, because $\mathit{reflexivity}$ implies $\mathit{seriality}$, but we include both for clarity. We call it a $\mathit{crash}$ because it renders the logic $\mathcal{L}$ unsound.

\begin{theorem}\label{thm:crash}
	A logic $\mathcal{L}$ that crashes is unsound.
\end{theorem}
\begin{proof}[Proof.]
	Suppose $\mathcal{L}$ $\mathbf{crashes}$. Then $\mathcal{L}$'s axioms derive L\"ob's Theorem, and there is one $\Box_{i}$ operator with relation $\mathcal{R}_i$ that is serial. A $\mathit{serial}$ frame relation corresponds to the formula $\Box_{i}\varphi\iimplies\Diamond_{i}\varphi$. Equivalently, $\tlnot\Box_{i}\varphi \tlor \tlnot \Box_{i} \tlnot\varphi$. By DeMorgan's Law, this is equivalent to $\tlnot(\Box_{i}\varphi \tland \Box_{i}\varphi)$. Because $\Box_{i}$ distributes and exports for conjunction, this is equivalent to $\tlnot\Box_{i}(\varphi\tland\tlnot\varphi)$, which amounts to the claim that contradictions are not accessible via the $\Rel{i}$ relation. This is equivalent to $\Box_{i}(\varphi\tland\tlnot\varphi)\iimplies(\varphi\tland\tlnot\varphi)$, because assuming the opposite of a theorem implies a contradiction. By the Rule of Necessitation, it follows that $\Box_{i}(\Box_{i}(\varphi\tland\tlnot\varphi)\iimplies(\varphi\tland\tlnot\varphi))$. But by L\"ob's Theorem, it follows from this that $\Box_{i}(\varphi\tland\tlnot\varphi)$. This contradicts with the earlier result, that $\tlnot\Box_{i}(\varphi\tland\tlnot\varphi)$, which is an equivalent expression of seriality. Therefore, the $\mathbf{crashing}$ logic $\mathcal{L}$ is unsound.
\end{proof}

The prevailing theories in epistemic logic ignore L\"ob's Obstacle. Due to the prevalence of the Positive and Negative Introspection axioms, and serial frame relations on the knowledge or belief operators, these logics crash. In order to avoid the crash, they must deny the existence of L\"ob sentences in the language. But this so limits the expressive power of the logic that it does not reasonably apply to human-like agents.

In what follows, we round up the usual suspects of epistemic-doxastic logic and show that, on the assumption that they aim to capture human-like reasoning, they crash into L\"ob's Obstacle. Both epistemic logics and doxastic logics, as they are commonly axiomatized, result in contradictions.

\section{Epistemic Logics that Crash}
\label{sec:crashing_logics}
\subsection{\SFive\ Epistemic Logic}
The most prominent epistemic logic in the literature, by far, is \SFive\ epistemic logic. \SFive\ epistemic logic is routinely presented as the logic of knowledge, and often serves as a static base for dynamic extensions to epistemic logic involving action and communication. Its characteristic axioms are:

\begin{table}[H]
	\begin{center}
		\begin{tabular}{| l r |}
			\hline
			$\Kns{i}(\varphi \iimplies \psi)\iimplies (\Kns{i}\varphi \iimplies \Kns{i}\psi)$ & Axiom K \\
			$\Kns{i}\varphi \iimplies \varphi$ & Truth Axiom (Axiom T)\\
			$\tlnot\Kns{i}\varphi \iimplies \Kns{i}\tlnot\Kns{i}\varphi$ & Negative Introspection (Axiom 5)  \\
			From $\vdash \varphi$ and $\vdash \varphi \iimplies \psi$, infer $\vdash\psi$ & Modus Ponens\\
			From $\vdash \varphi$, infer $\vdash \Kns{i}\varphi$ & Necessitation of $\Kns{i}$\\
			\hline
		\end{tabular}
		\caption{Logic of \SFive}~\label{S5}
	\end{center}
\end{table}

Axiom 5 is called the Negative Introspection axiom, or sometimes in philosophy circles, the Wisdom Axiom. It is read, ``If $i$ does not know that $\varphi$, then $i$ knows that $i$ doesn't know $\varphi$". Other than being clearly invalid for humans, this axiom and (3) allows us to derive,
\begin{equation*}
\Kns{i}\varphi \iimplies \Kns{i}\Kns{i}\varphi
\end{equation*}

\textbf{Proof:}

\begin{tabular}{ll}
\toprule
(1) & $\lnot K_i\lnot K_i\varphi \rightarrow K_i\varphi$ \dotfill Contrapositive of Axiom 5 \\
(2) & $K_i\lnot K_i\lnot K_i\varphi \rightarrow K_iK_i\varphi$ \dotfill Rule of Necessitation on (1), Axiom K \\
(3) & $\varphi \rightarrow \lnot K_i\lnot \varphi$ \dotfill Axiom T, Contrapositive \\
(4) & $\lnot K_i\lnot\varphi \rightarrow K_i\lnot K_i\lnot\varphi$ \dotfill Axiom 5 \\
(5) & $\varphi \rightarrow K_i\lnot K_i\lnot \varphi$ \dotfill (3), (4) \\
(6) & $K_i\varphi \rightarrow K_i\lnot K_i\lnot K_i\varphi$ \dotfill $K_i\varphi$/$\varphi$, (5) \\
(7) & $K_i\varphi\rightarrow K_iK_i\varphi$ \dotfill (2), (6) \\
\bottomrule
\end{tabular}

\hfill $\mathcal{QED}$
\\
Thus, \SFive\ satisfies L\"ob's three conditions, if we assume the presence of self-referential sentences possible, which we should. Therefore, with $\Kns{i}$ instead of $\Box$, the proof of L\"ob's Theorem is possible in this brand of \SFive. However, to make matters worse, the antecedent of L\"ob's Theorem is itself an axiom of S5. Therefore, $\Kns{i}\varphi$ is a theorem, for all $\varphi$.

We take this as a \emph{reductio ad absurdum} that \SFive\ epistemic logic cannot be a logic for reasoning about the knowledge of agents with expressive power beyond Peano arithmetic. Therefore, it cannot be a logic of knowledge for humans, or human-like agents.

\subsection{Hintikka's S4 Epistemic Logic}
\label{sec:hint_s4}
In Hintikka's 1967 \emph{Knowledge and Belief: A logic of the two notions}, he presented an epistemic logic for determining the validity and consistency of claims people make about knowledge and belief. A formal epistemic theory should strive for adequate philosophical grounding in good epistemology.  Hintikka dedicates a great portion of the book to exploring how his formal system handles the intuitive judgments of philosophers regarding ordinary language statements, which was the primary method at the time. An epistemic logic divorced from a philosophical foundation is no longer an epistemic logic for reasoning about human-like knowledge.

He rejected out of hand the negative introspection axiom (Axiom 5) for knowledge, but chose to include positive introspection (Axiom 4):\newline $\Kns{i}\varphi\iimplies\Kns{i}\Kns{i}\varphi$. His interpretation of $\Kns{i}$ is "$i$ could come to know $\varphi$ based on what $i$ currently knows". This subjunctive or hypothetical reading of knowledge offers some intuitive appeal, and could justify including Axiom 4. 

However, since Hintikka's epistemic system is meant for human-like reasoners who can express sentences like, ``If I know this sentence is true, then 1 + 1 = 2," it must be able to soundly handle self-referential sentences. Likewise, including Axiom 4, Positive Introspection, means that Hintikka's logic satisfies the L\"ob Conditions. Thus, it derives L\"ob's Theorem. Hintikka includes the Truth Axiom, $\Kns{i}\varphi \iimplies \varphi$, from which it follows that Hintikka's epistemic logic crashes.

\begin{theorem}
	Hintikka's $\mathcal{S}\mathit{4}$ epistemic logic crashes.
\end{theorem}
\begin{proof}[Proof.]
	The logic $\mathcal{S}\mathit{4}$ is a normal modal logic with a transitive frame relation. Epistemic logic for humans expresses reflective sentences. Thus, L\"ob's Theorem is derivable. Additionally, the frame relation for $\mathcal{S}\mathit{4}$ is reflexive. Therefore, by definition \ref{Crash}, $\mathcal{S}\mathit{4}$ crashes.
\end{proof}

\subsection{Kraus and Lehmann System}
\label{sec:kl}

In \cite{KrausLehmann}, Kraus and Lehmann to combine knowledge and belief in a single system of modal logic suitable for human-like agents. In particular, they extend the work of Halpern and Moses in \cite{HalpernMoses}, who conjecture that such a multimodal logic would be useful for modeling agents with incomplete information. As an early attempt at formalizing knowledge and belief together in a single modal axiom scheme, it represents an important milestone.

They axiomatize knowledge and belief as follows.

\begin{table}[H]
	\begin{center}
		\begin{tabular}{| l r |}
			\hline
			$\Kns{i}(\varphi \iimplies \psi) \iimplies (\Kns{i}\varphi \iimplies \Kns{i}\psi)$ & Distribution of $\Kns{i}$ \\
			$\Kns{i}\varphi \iimplies \varphi$ & Truth \\
			$\tlnot\Kns{i}\varphi \iimplies \Kns{i}\tlnot\Kns{i}\varphi$ & Negative Introspection \\
			$\Bels{i}(\varphi \iimplies \psi) \iimplies (\Bels{i}\varphi \iimplies \Bels{i}\psi)$ & Distribution of $\Bels{i}$\\
			$\Bels{i}\varphi \iimplies \BPoss{i}\varphi$ & Belief Consistency \\
			$\Kns{i}\varphi \iimplies \Bels{i}\varphi$ & Knowledge Entails Belief \\
			$\Bels{i}\varphi \iimplies \Kns{i}\Bels{i}\varphi$ & Conscious Belief\\
			From $\vdash \varphi$ and $\vdash \varphi \iimplies \psi$, infer $\vdash\psi$ & Modus Ponens\\
			From $\vdash \varphi$, infer $\vdash \Kns{i}\varphi$ & Necessitation of $\Kns{i}$\\
			\hline
		\end{tabular}
		\caption{Logic of Kraus and Lehmann}~\label{KL}
	\end{center}
\end{table}

In their article, and in Meyer and van der Hoek's \cite{MeyerHoek}, they show that $\Bels{i}(\Bels{i}\varphi \iimplies \varphi)$ is a theorem. This, combined with the satisfaction of the L\"ob Conditions, entails that Kraus and Lehmann's logic crashes. 

\begin{theorem}
	Kraus and Lehmann's logic crashes.
\end{theorem}
\begin{proof}[Proof.]
	The logic consists of two normal modal operators, $\Kns{i}$ and $\Bels{i}$. For the $\Kns{i}$ operator, the Truth and Negative Introspection axioms together entail Positive Introspection. Therefore, $\Kns{i}$ is an \SFive\ operator, and crashes. For $\Bels{i}$, we have a serial frame relation due to the Belief Consistency axiom, and we can derive Positive Belief Introspection from Conscious Belief and Knowledge Entails Belief. Thus, $\Bels{i}$ has positive introspection and is a normal modal operator, which expresses reasoning about reflective sentences. Thus, L\"ob's Theorem is derivable for $\Bels{i}$. Because $\Bels{i}$ has a serial frame relation, it crashes. Furthermore, since $\Bels{i}(\Bels{i}\varphi \iimplies \varphi)$ is a theorem, it follows that $\Bels{i}\varphi$ is a theorem.
\end{proof}

\subsection{\KDFourFive\ Doxastic Logic}
\label{sec:crap}
\KDFourFive\ is perhaps the most dominant formalization of doxastic logic. It includes Positive and Negative Belief Introspection and the Belief Consistency axioms. 

\begin{table}[H]
	\begin{center}
		\begin{tabular}{| l r |}
			\hline
			$\Bels{i}(\varphi \iimplies \psi) \iimplies (\Bels{i}\varphi \iimplies \Bels{i}\psi)$ & Distribution of $\Bels{i}$\\
			$\Bels{i}\varphi \iimplies \BPoss{i}\varphi$ & Belief Consistency \\
			$\Bels{i}\varphi \iimplies \Bels{i}\Bels{i}\varphi$ & Positive Belief Introspection \\
			$\tlnot\Bels{i}\varphi \iimplies \Bels{i}\tlnot\Bels{i}\varphi$ & Negative Belief Introspection\\
			From $\vdash \varphi$ and $\vdash \varphi \iimplies \psi$, infer $\vdash\psi$ & Modus Ponens\\
			From $\vdash \varphi$, infer $\vdash \Kns{i}\varphi$ & Necessitation of $\Bels{i}$\\
			\hline
		\end{tabular}
		\caption{Logic of \KDFourFive}~\label{KD45}
	\end{center}
\end{table}

\KDFourFive\ satisfies three L\"ob Conditions (Axiom K - Distribution of $\Bels{i}$, Axiom 4 - Positive Belief Introspection, and the Rule of Necessitation), so for agents capable of self-referential reasoning, L\"ob's Theorem is derivable.

The Belief Consistency Axiom $\Bels{i}\varphi\iimplies\BPoss{i}\varphi$ is equivalent to\newline $\tlnot(\Bels{i}\varphi \tland \Bels{i}\tlnot\varphi)$, which is furthmore equivalent to $\tlnot\Bels{i}(\varphi \tland \tlnot \varphi)$, which results in the following disaster.

\begin{theorem}[Consistency Disaster]~\label{no_bel_cons}
	If $\tlnot\Bels{i}(\varphi \tland \tlnot \varphi)$ and\\ $\Bels{i}(\Bels{i}\varphi \iimplies\varphi)\iimplies \Bels{i}\varphi$ are theorems, then $\Bels{i}(\varphi \tland \tlnot \varphi)$ is a theorem.
\end{theorem}
\textbf{Proof:}

\begin{tabular}{ll}
\toprule
(1) & $\lnot B_i(\varphi \land \lnot \varphi)$ \dotfill Belief is Consistent \\
(2) & $B_i(\varphi\land\lnot\varphi)\rightarrow (\varphi \land \lnot \varphi)$ \dotfill (1), logically equivalent \\
(3) & $B_i(B_i(\varphi \land \lnot \varphi)\rightarrow (\varphi \land \lnot \varphi))$ \dotfill Necessitation of $B_i$, (2) \\
(4) & $B_i(B_i(\varphi \land \lnot \varphi) \rightarrow (\varphi \land \lnot \varphi))\rightarrow B_i(\varphi \land \lnot \varphi)$ \dotfill L\"ob's Theorem \\
(5) & $B_i(\varphi \land \lnot \varphi)$ \dotfill (3), (4) \\
\bottomrule
\end{tabular}

\hfill $\mathcal{QED}$
\\

Thus, with Belief Consistency, L\"ob's Theorem, and Theorem \ref{no_bel_cons}, an inconsistency follows. Note the similarity here to the problem facing PA$^+$ were it to prove its own consistency, due to G\"odel's Second Incompleteness Theorem.\cite{godel}

\subsection{Analysis}
Defenders of these logics may wish to deny that these logics can express reflective sentences. However, as we argued earlier, this renders them unsuitable for human-like agents. Furthermore, in many cases, reflective sentences in the form of fixpoints are explicitly included. For example, notions of common knowledge and common belief are frequently added to one of the above logics. These include theorems of the form $\mathbf{C} \varphi \equiv (\varphi \tland \Lambda_{i\in\mathbb{G}}\Kns{i}\mathbf{C}\varphi)$, a straightforwardly reflective sentence, often parsed as ``$\varphi$, and everybody knows this sentence". 

Clearly, the solution is to remove one of the problematic axioms. One might wonder whether it would be acceptable to abandon the Truth Axiom for knowledge, or the Belief Consistency Axiom for belief, and allow L\"ob's Theorem to hold for the knowledge or belief operator in a way that avoids inconsistency. This would introduce more modesty to the notion of knowledge, where a human-like agent knows that her knowledge is true only for those propositions that she actually knows, but not in the general sense.\footnote{Smullyan referred to such agents as modest agents, for they are confident in the accuracy only of particular beliefs that they have good reasons (proofs) for, but lack a general confidence in their own beliefs.} 

What would this mean for epistemology? A false proposition would no longer imply a lack of knowledge, and the rejection of the truth axiom goes against the entire history of Western philosophical thought. This is to say, it would require robust philosophical defense, which we are not prepared to give here. Relaxing the Truth Axiom allows positive and negative introspection to live harmoniously with self-reference. We leave this for future work to explore. We note here, however, that work in \cite{modal_prisoner} show that agents with L\"ob's Theorem holding for their epistemic operators are able to cooperate in the Prisoner's Dilemma game.\footnote{Specifically, they are programs capable of inspecting their own source code.} Related research in formal agent modeling via programs that play games with each other and can examine each other's source code has been explored by Binmore, Howard, McAfee, Tennenholtz, and other game theorists. Barasz \etal introduce the novel approach of using the so-called G\"odel-L\"ob ($\mathbb{GL}$) modal logic of provability, presented in detail by Boolos in \cite{Boolos}. This logic consists of a $\Box$ operator for ``is provable in PA", with the above L\"ob Theorem as an axiom, in addition to Axiom K (Distribution).

Similarly, relaxing the Belief Consistency Axiom reduces a doxastic logic from a normative system for correct reasoning to a merely descriptive system about the psychology of agents with sometimes contradictory beliefs. Rather than pursuing that route, we explore axiom schemas that avoid the derivation of L\"ob's Theorem entirely.

\section{Avoiding L\"ob}
\label{sec:avoiding_lob}
In order to avoid L\"ob's Obstacle, we model agents with the following properties. First, their knowledge is true. Second, they do not always know whether they know something or not. Third, they do not believe inconsistencies. Fourth, they have strong evidence for their beliefs. This last condition represents a key weakening, and must be spelled out in greater detail. This model is better suited to human-like agents who interact with reality and lack perfect information about the environment, and most importantly, a logic for agents with these properties is what we call L\"ob Safe.

\begin{definition}
	A logic $\mathcal{L}$ with modal operators $\mathbf{\Box_{{i}\in\mathbb{N}}}$ is $\textbf{L\"ob Safe}$ just in case for each $\Box_{i}$ some L\"ob Condition is false or $\Box_{i}$ is not defined by a serial frame relation $\Rel{i}$. 
\end{definition}

Recall the L\"ob Conditions are the expressibility of L\"ob Sentences, $\mathcal{L}$'s being a normal modal logic, and the inclusion of Axiom 4 (which corresponds to a transitive frame relation). Thus, in defining a multimodal logic of agency that is L\"ob Safe, we must construct it with operators that carefully navigate these conditions. The most straightforward approach is to make sure a modal operator is never both transitive and serial. We adopt this approach with each modal operator as follows.

First, we reduce the belief operator to a basic System D modality, consisting of Belief Consistency: $\Bels{i}\varphi \iimplies \BPoss{i}
\varphi$. Second, we include an axiom imposing a necessary condition on belief, $\Bels{i}\iimplies \BPoss{i}\Kns{i}\varphi$, which we call Reasonable Belief (RB). Read this as "$i$ believes that $\varphi$ only if it is reasonable for $i$ that $i$ knows $\varphi$."\footnote{Our parsing of $\BPoss{i}$ as ``reasonable" is based on the notion that ``reasonable" connotes an existential quantification over a reason relation, e.g. "there is some reason to believe..."} A stronger version of this condition, $\Bels{i}\varphi \iimplies \Bels{i}\Kns{i}\varphi$, which we call Overconfident Belief (OB), is sufficient for deriving Positive Introspection for Belief (in conjunction with the axiom Knowledge implies Belief), which would cause L\"ob's Theorem to be derivable. 

Roughly, the difference between RB and OB is the strength of evidence $i$ must have. RB states that there must be some reason to believe, from $i$'s perspective, that $i$ knows $\varphi$. OB states that there must be no reason to believe, from $i$'s perspective, that $i$ does not know $\varphi$, which is a taller order, resulting in a situation in which $i$ believes that all of their beliefs constitute knowledge. A realistic rational agent should not hold this belief, just out of epistemic caution or modesty.

We do not include Negative Introspection for Belief because including it allows us to derive Positive Introspection for Belief. 

For the knowledge operator, we include the axiom that Knowledge entails Truth, as well as the axiom that Knowledge entails Belief, both of which correspond to most epistemological views. Both $\Bels{i}$ and $\Kns{i}$ remain normal modal operators.

The resulting logic is L\"ob Safe while remaining realistic for human-like agents, and is in the class of logics whose completeness is established via Sahlqvist's Theorems. We call this resulting logic \LSED, for Reasonable L\"ob-Safe Epistemic Doxastic logic.

\begin{table}[H]
	\begin{center}
		\begin{tabular}{| l r |}
			\hline
			$\Kns{i}(\varphi \iimplies \psi) \iimplies (\Kns{i}\varphi \iimplies \Kns{i}\psi)$ & Distribution of $\Kns{i}$ \\
			$\Kns{i}\varphi \iimplies \varphi$ & Truth \\
			$\Bels{i}(\varphi \iimplies \psi) \iimplies (\Bels{i}\varphi \iimplies \Bels{i}\psi)$ & Distribution of $\Bels{i}$\\
			$\Bels{i}\varphi \iimplies \BPoss{i}\varphi$ & Belief Consistency \\
			$\Kns{i}\varphi \iimplies \Bels{i}\varphi$ & Knowledge entails Belief \\
			$\Bels{i}\varphi \iimplies \BPoss{i}\Kns{i}\varphi$ & Reasonable Belief\\
			From $\vdash \varphi$ and $\vdash \varphi \iimplies \psi$, infer $\vdash\psi$ & Modus Ponens\\
			From $\vdash \varphi$, infer $\vdash \Kns{i}\varphi$ & Necessitation of $\Kns{i}$\\
			\hline
		\end{tabular}
		\caption{Logic of \LSED}~\label{lsed}
	\end{center}
\end{table}

Neither the belief operator nor the knowledge operator is susceptible to L\"ob's Obstacle, as L\"ob's Theorem is not derivable in the system.

\begin{theorem}
	\LSED\ is L\"ob Safe.
\end{theorem}
\begin{proof}[Proof.]
	We show that the $\Bels{i}$ operator does not satisfy Axiom 4, or Positive Belief Introspection, as a theorem, meaning some L\"ob Condition is false for it.

	In the counterexample below, we see that the belief operator lacks positive introspection, where $\Bels{i}$ is the belief operator and $\Rel{b}^i$ is the relation defining accessibility for belief.

\begin{figure}[H]
	\begin{center}
		\begin{tikzpicture}[->,>=stealth',shorten >=1pt,auto,node distance=3cm,
		thick,base node/.style={circle,draw,minimum size=35pt}]
		
		\node[base node] (w) {\begin{tabular}{c}
			$w:  p$ \\ $\Bels{i}p$ \\ $\tlnot\Bels{i}\Bels{i}p$
			\end{tabular}};
		\node[base node] (v) [right of=w] {\begin{tabular}{c}
			$v:  p$ \\ $\tlnot\Bels{i}p $
			\end{tabular}};
		\node[base node] (u) [right of=v] {$u: \tlnot p$};
		\path[]
		(w) edge node[above] {$\Rel{b}^i$} (v)
		(v) 
		edge node[above] {$\Rel{b}^i$} (u)
		(u) edge [<-, loop above] node {$\Rel{b}^i$} (u);
		
		\end{tikzpicture}
	\end{center}
	\caption{A counterexample to $\Bels{i}\varphi\iimplies\Bels{i}\Bels{i}\varphi$.}
\end{figure}
\end{proof}

The Sahlqvist class of modal formulas are described for their formal properties, and if a logic is axiomatized only by Sahlqvist formulas, then it is sound and complete with respect to first order frame relations corresponding to each axiom.\cite{sahlqvist} The axioms for \LSED\ are each Sahlqvist formulas, so \LSED\ is sound and complete with respect to frames defined by the corresponding first order formulas. We show the translation from Reasonable Belief to the corresponding first order frame condition. Since the first order frame conditions of the other axioms are well known, this suffices to show that \LSED\ is sound and complete.

\emph{Soundness and Completeness}. In proving soundness and completeness we apply the Sahlqvist-van Benthem Algorithm, which is described in detail by Blackburn in \cite{modal}.

$\Bels{i}\varphi \iimplies \BPoss{i}\Kns{i}\varphi$.
\begin{align*}
	&\rightsquigarrow \forall P,y, (\Rel{b}^i(x,y)\Longrightarrow P(y)) \Longrightarrow \exists z, \forall z',(\Rel{b}^i(x,z)\tland (\Rel{k}^i(z,z') \Longrightarrow P(z')))\\
	&\rightsquigarrow \forall y, (\Rel{b}^i(x,y)\Longrightarrow \lambda u.(\Rel{b}^i(x,u))(y)) \\ &\ \ \ \ \Longrightarrow \exists z,\forall z', (\Rel{b}^i(x,z)\tland (\Rel{k}^i(z,z') \Longrightarrow \lambda u.(\Rel{b}^i(x,u))(z')))\\
	&\rightsquigarrow \forall y, (\Rel{b}^i(x,y)) \Longrightarrow \exists z,\forall z', (\Rel{b}^i(x,z)\tland (\Rel{k}^i(z,z') \Longrightarrow  (\Rel{b}^i(x,z'))))		
\end{align*}

It is an unfamiliar frame condition, but it is a first order formula, and because the Sahlqvist Correspondence Theorem establishes correspondence between the modal formula and the first order formula, it entails that the class of frames defined by the first order formula as a frame condition are those for which the modal formula are sound. The antecedent is satisfied for reflexive and serial frames, so we can focus on the consequent to get an intuition about the condition. It says that the compose of $\Rel{k}^i \circ \Rel{b}^i$ is a subset of the $\Rel{b}^i$ relation, meaning $i$'s beliefs are constrained to propositions that seem possibly known. She can reflect on her beliefs and always consider it reasonable that she knows. This functions as a \textit{constraint} on her belief.
%
%

We have shown that \LSED\ lacks one of the L\"ob Conditions, and is therefore L\"ob Safe, but we provide the following counterexample to L\"ob's Theorem for good measure.

\begin{figure}[H]
	\begin{center}
		\begin{tikzpicture}[->,>=stealth',shorten >=1pt,auto,node distance=4.5 cm,
		thick,base node/.style={circle,draw,minimum size=35pt}]
		
		\node[base node] (w) {\begin{tabular}{c}
			$w:  \Bels{i}(\Bels{i}p \iimplies p)$ \\ $\tlnot\Bels{i}p$
			\end{tabular}};
		\node[base node] (v) [right of=w] {\begin{tabular}{c}
			$v:  \Bels{i}p \iimplies p $ \\ $\tlnot p $
			\end{tabular}};
		\path[]
		(w) edge node[above] {$\Rel{b}^i$} (v)
		(v) 
		edge [<-, loop above] node {$\Rel{b}^i$} (v);
		
		\end{tikzpicture}
	\end{center}
	\caption{A counterexample to $\Bels{i}(\Bels{i}\varphi \iimplies \varphi) \iimplies \Bels{i}\varphi$.}
\end{figure}

We assume in world $w$ that $\Bels{i}(\Bels{i}p \iimplies p)$ holds. From this it follows that in all worlds accessible via the $\Rel{b}^i$, \emph{e.g.} $v$, $\Bels{i}p \iimplies p$ holds. This holds for $v$ when it has reflexive access only to itself, and $\tlnot p$ is the case. Because $\tlnot p$ is the case $v$, $\tlnot \Bels{i}p$ is the case at $w$, concluding the counterexample.

The logic also includes a sort of weakened positive introspection theorem about knowledge, from the Knowledge implies Belief axiom and Weakly Reasonable Belief. It is $\Kns{i}\varphi \iimplies \Poss{i}\Kns{i}\varphi$, which we read as ``if $i$ knows that $\varphi$ then it is possible for $i$ given the evidence that $i$ knows $\varphi$". This seems intuitive, and does not allow L\"ob's theorem to destroy the integrity of knowledge. It is perhaps a satisfying compromise for those who find positive introspection about knowledge to be intuitive. The contrapositive of this weak positive introspection formula, $\Kns{i}\tlnot\Kns{i}\varphi \iimplies \tlnot\Kns{i}\varphi$ is an instance of the T axiom (Truth axiom), so it turns out to have been a theorem all along anyway, for any epistemic logic with the Truth axiom for knowledge.

An alternative L\"ob Safe logic uses $\Bels{i}\varphi \iimplies \Poss{i}\Kns{i}\varphi$ as a relaxed axiom, which might be called Supported Belief, based on its relating of belief explicitly to evidence of knowledge. We can state this as ``$i$ believes $\varphi$ only if some evidence for $i$ supports that $i$ knows $\varphi$. The relationship between ``seemingly possible" and ``possible given the evidence" is outside of the present scope, but we briefly note that the claims we are here committed to are that ``seemingly possible" entails ``evidentially possible", which may intuitively function as a constraint on agents to base perceptions on reasons and evidence, hence our names for Reasonable Belief and Supported Belief. Supported Belief is weaker than Reasonable Belief.

\begin{theorem}[Reasonable Belief Implies Supported Belief]\label{wser_wer}
	$\\ (\Bels{i}\varphi \iimplies \BPoss{i}\Kns{i}\varphi)
	\iimplies (\Bels{i}\varphi \iimplies \Poss{i}\Kns{i}\varphi).$
\end{theorem}
\begin{proof}
	It suffices to show that $\BPoss{i}\Kns{i}\varphi \iimplies \Poss{i}\Kns{i}\varphi$. This is an instance of Knowledge implies Belief, contraposed.
\end{proof}

This weaker logic, which we call \LSEDmin, for Supported L\"ob-Safe Epistemic Doxastic logic, is axiomatized as follows.

\begin{table}[H]
	\begin{center}
		\begin{tabular}{| l r |}
			\hline
			$\Kns{i}(\varphi \iimplies \psi) \iimplies (\Kns{i}\varphi \iimplies \Kns{i}\psi)$ & Distribution of $\Kns{i}$ \\
			$\Kns{i}\varphi \iimplies \varphi$ & Truth \\
			$\Bels{i}(\varphi \iimplies \psi) \iimplies (\Bels{i}\varphi \iimplies \Bels{i}\psi)$ & Distribution of $\Bels{i}$\\
			$\Bels{i}\varphi \iimplies \BPoss{i}\varphi$ & Belief Consistency \\
			$\Kns{i}\varphi \iimplies \Bels{i}\varphi$ & Knowledge entails Belief \\
			$\Bels{i}\varphi \iimplies \Poss{i}\Kns{i}\varphi$ & Supported Belief\\
			From $\vdash \varphi$ and $\vdash \varphi \iimplies \psi$, infer $\vdash\psi$ & Modus Ponens\\
			From $\vdash \varphi$, infer $\vdash \Kns{i}\varphi$ & Necessitation of $\Kns{i}$\\
			\hline
		\end{tabular}
		\caption{Logic of \LSEDmin}~\label{lsedmin}
	\end{center}
\end{table}

\begin{theorem}
	\LSEDmin\ is L\"ob Safe.
\end{theorem}
\begin{proof}[Proof.]
	Due to Theorem \ref{wser_wer}, the set of theorems of \LSEDmin\ is a subset of those of \LSED. \LSED\ is L\"ob Safe, so \LSEDmin\ is as well.
\end{proof}

\emph{Soundness and Completeness}. We once again apply the Sahlqvist-van Benthem Algorithm to generate the frame condition corresponding to Supported Belief.

	$\Bels{i}\varphi \iimplies \Poss{i}\Kns{i}\varphi$.
\begin{align*}
&\rightsquigarrow \forall P,y, (\Rel{b}^i(x,y)\Longrightarrow P(y)) \Longrightarrow \exists z, \forall z',(\Rel{k}^i(x,z)\tland (\Rel{k}^i(z,z') \Longrightarrow P(z')))\\
&\rightsquigarrow \forall y, (\Rel{b}^i(x,y)\Longrightarrow \lambda u.(\Rel{k}^i(x,u))(y)) \\ &\ \ \ \ \Longrightarrow \exists z,\forall z', (\Rel{k}^i(x,z)\tland (\Rel{k}^i(z,z') \Longrightarrow \lambda u.(\Rel{b}^i(x,u))(z')))\\
&\rightsquigarrow \forall y, (\Rel{b}^i(x,y)) \Longrightarrow \exists z,\forall z', (\Rel{k}^i(x,z)\tland (\Rel{k}^i(z,z') \Longrightarrow  (\Rel{b}^i(x,z'))))		
\end{align*}

This formula captures an analogous class of frames where it is the compose of $(\Rel{k}^i \circ \Rel{k}^i)$ that is a subset of $\Rel{b}^i$. Intuitively, this constrains $i$'s beliefs to propositions that are one epistemic step away from knowledge. She has evidence, though inconclusive, that her beliefs constitute knowledge.

\section{An Alternative Approach}
In \cite{Critch16}, Critch develops a solution to the L\"obian Obstacle that leans into L\"ob. He identifies a version of L\"ob's Theorem that is parametrically bounded in proof length by the length of an input program. His purpose is to leverage L\"ob's Theorem in the context of program-agents playing Prisoner's Dilemmas against each other, with access to each others source code. In this context, a certain type of program-agent can be defined that can systematically cooperate with other program-agents, including itself, for which it can prove that the other program-agent will likewise cooperate. 

This use of L\"ob's Theorem represents an intriguing approach to achieving cooperation in the Prisoner's Dilemma. The modality in this case is that of provability logic, the previously mentioned $\mathbb{GL}$ with parametric bounds introduced. L\"ob's Theorem is valid, and other properties like Truth and Belief Consistency are invalid. By avoiding the Truth Axiom, provability logic soundly includes L\"ob's Theorem, and thus avoids crashing. Critch demonstrates its usefulness in developing agents that cooperate. This modality models an agent that is not similar to a rational human interacting with the external world. Rather, this agent is firmly embedded in the world of mathematical knowledge. So, we can consider our approach to be one that complements the work done by Critch, offering an alternative model of agent epistemology that avoids crashing into the L\"obian Obstacle. Whether game theoretic agents with one of our L\"ob Safe logics as a foundation would cooperate or defect in the Prisoner's Dilemma remains the subject of future research.

\section{Conclusion}
We have presented a problem facing agents capable of reflective reasoning and explained how the L\"ob Conditions (Axiom K, Axiom 4, and Rule of Necessitation), in conjunction with reflective sentences, derive L\"ob's Theorem. We showed that standard assumptions about knowledge and belief cause the resulting logics to crash into L\"ob's obstacle, for both knowledge operators and belief operators, resulting in inconsistency. These logics include \SFive, \KDFourFive, $\mathcal{S}\mathit{4}$, and Kraus and Lehmann's epistemic doxastic logic. Most models of agency in game theory, computer science, and philosophy, involve one of these logics as a foundational element, and therefore are unsuitable for modeling reflective agents.

We responded to this by presenting \LSED\ and \LSEDmin, which avoid L\"ob's Obstacle while maintaining a defensible model of human-like knowledge and belief, suitable for rational agents interacting with the external world. 


\bibliographystyle{spbasic}      
\bibliography{bibliography}   

%
%

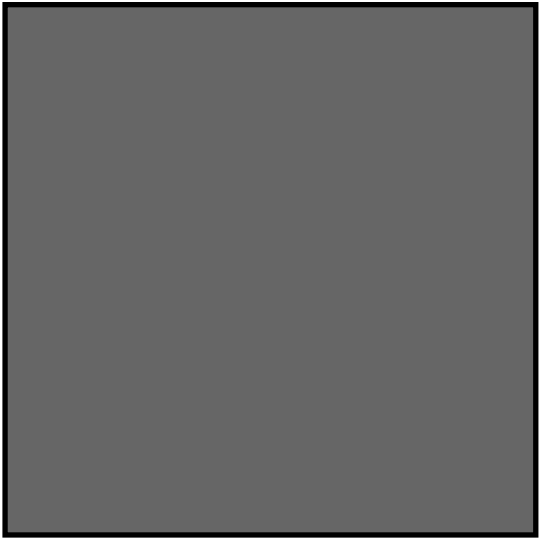
\end{document}